\documentclass[letterpaper, 10 pt, conference]{IEEEtran}

\makeatletter
\newcommand\semihuge{\@setfontsize\semihuge{22.3}{22}}
\makeatother

\usepackage{tabularx}

\usepackage{algpseudocode}
\usepackage{algorithm}
\usepackage{algorithmicx}

\usepackage{lipsum} % For algorithm smape after or before:
\usepackage{arydshln}
\usepackage[dvips]{color}
\usepackage{comment}
\usepackage{todonotes}
\usepackage{epsf}
\usepackage{epsfig}
\usepackage{times}
\usepackage{epsfig}
\usepackage{graphicx}
\usepackage{bbold}
\usepackage{mathtools}
\usepackage{mathrsfs}
\usepackage{amssymb}
\usepackage{pdfpages}
\usepackage{epstopdf}
\newfloat{algorithm}{t}{lop}

\usepackage{amsmath}
\usepackage{dsfont}
\usepackage{lettrine} % \lettrine[findent=1pt]{{{R}}}{}

\usepackage{amsmath,epsfig,amssymb,algorithm,algpseudocode,amsthm,cite,url}
\usepackage{subcaption}
\allowdisplaybreaks
\usepackage{csquotes}
\usepackage{verbatim}
\usepackage[english]{babel}
\usepackage{amsmath,amssymb}

\usepackage[justification=centering]{caption}
\usepackage{verbatim}

\newtheorem{theorem}{\bf Theorem}

\begin{document}
	\title{\semihuge{Deep Learning-Based Dynamic Watermarking for Secure Signal Authentication in the Internet of Things}}
	\IEEEoverridecommandlockouts
	\author{\IEEEauthorblockN{Aidin Ferdowsi\IEEEauthorrefmark{1} and Walid Saad\IEEEauthorrefmark{1}\\}
		\IEEEauthorblockA{\IEEEauthorrefmark{1}
			Wireless@VT, Bradley Department of Electrical and Computer Engineering, \\ Virginia Tech, Blacksburg, VA, USA,
			Emails: \{aidin,walids\}@vt.edu\\}\vspace{-1cm}
		\thanks{This research was supported by the U.S. National Science Foundation under Grants OAC-1541105 and CNS-1524634.}
	}
	\maketitle
	
	\maketitle
	

	% \fi
	\IEEEpeerreviewmaketitle
	
\begin{abstract}	
Securing the Internet of Things (IoT) is a necessary milestone toward expediting the deployment of its applications and services. In particular, the functionality of the IoT devices are extremely dependent on the reliability of their message transmission. Cyber attacks such as data injection, eavesdropping, and man-in-the-middle threats can lead to security challenges. Securing IoT devices against such attacks requires accounting for their stringent computational power and need for low-latency operations. In this paper, a novel deep learning method is proposed for dynamic watermarking of IoT signals to detect cyber attacks. The proposed learning framework, based on a long short-term memory (LSTM) structure, enables the IoT devices to extract a set of stochastic features from their generated signal and dynamically watermark these features into the signal. This method enables the IoT's cloud center, which collects signals from the IoT devices, to effectively authenticate the reliability of the signals. Furthermore, the proposed method prevents complicated attack scenarios such as eavesdropping in which the cyber attacker collects the data from the IoT devices and aims to break the watermarking algorithm. Simulation results show that, with an attack detection delay of under 1 second the messages can be transmitted from IoT devices with an almost $ 100\% $ reliability.
\end{abstract}
\section{Introduction}
The Internet of Things (IoT) will encompass a massive number of devices that must reliably transmit a diverse set of messages observed from their environment \cite{IoT2017Dawy}. The IoT will have applications in multiple areas including health care \cite{IoT2015Islam}, smart grids \cite{smartgrid2010Farhangi}, drones\cite{UAV2016Mozaffari} and industrial monitoring \cite{IoT2014Xu}. However, an effective deployment of these diverse IoT services near real-time, reliable, secure, and low complexity message transmission from the IoT devices.   

Most IoT architectures consist of four layers: perceptual, network, support, and application layer\cite{security2012Suo}. The perceptual layer is the most fundamental layer which collects all kind of information from the physical world using RFID, GPS, accelerometer, and gyroscope data. The other three layers are centered around communication and computation of IoT signals as well as personalized services for end users. Due to simplicity of the devices and components at the perceptual layer, their resource-constrained nature, and their low computational and storage capabilities, securing the IoT signals at this layer is notoriously challenging. 

Recently, a number of security solutions have been proposed for IoT signal protection \cite{Physical2015Mukherjee,physical2015Trappe,lightweight2014Lee,security2016Yaman}. The work in \cite{Physical2015Mukherjee} investigated physical layer security techniques in IoT scenario and presented
some methods that are suitable for protecting IoT applications. These physical layer security methods include optimal sensor censoring, channel-based bit flipping, probabilistic ciphering of quantized IoT signals, and artificial noise signal transmission. In \cite{physical2015Trappe}, the author suggested bridging the security gap in IoT devices by applying information theory and cryptography
at the physical layer of the IoT. An authentication protocol for the IoT is presented in \cite{lightweight2014Lee}, using lightweight encryption method in order to cope with constrained IoT devices. Moreover, in \cite{security2016Yaman}, authors developed a learning mechanisms for fingerprinting and authenticating IoT devices and their environment.

To provide further security to IoT-like cyber-physical systems (CPSs), the idea of watermarking has been studied in \cite{watermark2015Mo,Watermarking2017Satchidanandan,hespanhol2017dynamic,watermark2017Hosseini}. In \cite{watermark2015Mo}, a method was proposed to watermark a predefined signal unto a CPS input signal that enables detection of replay attacks in which an adversary repeats a sequence of past measurements. A dynamic watermarking algorithm was proposed in \cite{Watermarking2017Satchidanandan} for integrity attack detection in networked CPSs. In \cite{hespanhol2017dynamic}, the authors introduced a security scheme that ensures detection of attacks that add a nonzero-average-power signal to sensor measurements using a non-stationary watermarking technique. Finally, the authors in \cite{watermark2017Hosseini} analyzed the optimality of Gaussian watermarked signals against cyber attacks in linear time-variant IoT-like systems.

However, the solutions of \cite{Physical2015Mukherjee,physical2015Trappe,lightweight2014Lee,security2016Yaman} remain highly complex for deployment at the IoT perceptual layer, and require high computational power. Moreover, these methods do not take into account complex attacks such as eavesdropping in which the attacker collects data for a long time duration and uses it for designing undetectable attacks. Furthermore, the watermarking algorithms introduced in \cite{watermark2015Mo,Watermarking2017Satchidanandan,hespanhol2017dynamic,watermark2017Hosseini}, can be detrimental to the performance of the CPS since the augmented watermark is applied in parallel to the control signal of the system. This can, in turn, cause non-optimality in the performance of the system. In addition, the input signals to a CPS, in most cases, are physical signals such as a sensor or the mechanical power input to a generator, and therefore, are not applicable to IoT devices since altering the physical input to an IoT device such as thermometer requires changing the IoT devices' environment.

The main contribution of this paper is a deep learning framework for dynamic watermarking of IoT signals. This framework enables the IoT's \emph{cloud} to authenticate the reliability of signals and detect existence of a cyber attacker who seeks to degrade the performance of the IoT by changing the devices' output signal. The proposed deep learning algorithm uses long short-term memory (LSTM) \cite{chen2017machine} blocks to extract stochastic features such as spectral flatness, skewness, kurtosis, and central moments from IoT signal and watermarks these features inside the original signal. This dynamic feature extraction allows the cloud to detect sophisticated eavesdropping attacks, since the attacker will not be able to extract the watermarked information. Moreover, the proposed LSTM reduces complexity and latency of the attack detection compared to other security methods such as encryption. This allows LSTM to effectively complement the cryptographic and security solutions of an IoT. Simulation results show that, using the proposed approach and for under $ 1 $ second latency, the IoT signals can be reliably transmitted from IoT devices to the cloud.

The rest of the paper is organized as follows. Section \ref{sect:sensor} introduces spread spectrum watermarking for IoT attack detection. Section \ref{sect:DeepMark} presents the proposed learning algorithm. Section \ref{sect:sim} analyzes the simulation and conclusions are drawn in Section \ref{sect:conc}.
\section{Signal Authentication and Attack Detection Using Watermarking}\label{sect:sensor}
Consider an IoT device (IoTD) which generates a signal $ y(t) $ at time step $ t $ with sampling frequency $ f_s $, and transmits this signal to a cloud, that uses the received signal is used for estimation and control of the IoTD operation. In this system, we consider an adversary who seeks to compromise the IoTD or the communication link between the IoTD and cloud, and, then, manipulates the transmitted signal. In this case, the transmitted signal will be $ \bar{y}(t) \neq y(t) $ which will cause an estimation error at the cloud. Therefore, the cloud must implement an attack detection mechanism to differentiate between the honest and false signal received from an IoTD. Here, we propose a watermarking scheme for effective IoTD signal authentication and attack detection.
\subsection{Spread Spectrum Watermarking}
Watermarking uses a hidden, predefined non-perceptual code (bit stream) inside a multimedia signal to authenticate the ownership of such signals. One of the widely used watermarking methods is spread spectrum (SS) \cite{Improved2003Malvar} in which a key pseudo-noise sequence is added to the original signal. The watermarked signal can then be written as follows:
\begin{align}\label{watermarking}
	w(t)=y(t)+\beta b p(t) \quad \textrm{for } t=1,\dots,n,
\end{align}
where $ w $ is the watermarked signal, $ p $ is a pseudo-noise binary sequence of $ +1 $ and $ -1 $, $ \beta $ is the relative power of the pseudo-noise signal to the original signal, $ b $ is the hidden bit in the signal which can take values $ +1 $ and $ -1 $, and $ n $ is the number of samples (frame length) of the original signal used to hide a single bit. To extract the watermarked bit, the cloud receives the watermarked signal and correlates it with the key pseudo-noise sequence. The extraction process will be:
	\begin{align}\label{correlation}
		\tilde{b}&=\frac{<w,p>_{n}}{\beta n}=\frac{<y,p>_{n}+\beta b<p,p>_{n}}{\beta n}=\tilde{y}+b,
	\end{align} 
where, for $ \tilde{b}>0 $ the extracted bit is $ 1 $ and for $\tilde{b}<0 $ the extracted bit is $ -1 $ and $ <w,p>_{n} $ is the inner production of $ n $ samples of $ w $ and $ p $. $ p(t) $ and $ y(t) $ are independent stochastic variables at time $ t $ and we consider $ y(t) $ having mean $ \mu $ and variance $ \sigma^2 $. Next, we analyze the bit error rate of the extracted bit.
\begin{theorem} \label{Theorem:BER}
	In the proposed SS watermarking scheme the bit extraction error is $ \frac{1}{2}\textrm{erfc}\left(\frac{\beta\sqrt{n}}{\sigma \sqrt{2}}\right) $.
\end{theorem}
\begin{proof}
	For $ \tilde{y} $, we can write:
	\begin{align}
		\tilde{y}&=\frac{1}{\beta n} \sum_{t=1}^{n} y(t)p(t)\nonumber=\frac{1}{\beta n}\left( \sum_{t\in \mathcal{P}^+} y(t)-\sum_{t\in \mathcal{P}^-} y(t)\right)\\\label{sum}
		&=\frac{n_{+}}{\beta n}\frac{1}{n_{+}}\sum_{t\in \mathcal{P}^+} y(t)-\frac{n_{P^-}}{\beta n}\frac{1}{n_{-}}\sum_{t\in \mathcal{P}^-} y(t),
	\end{align}
	where $ \mathcal{P}^+=\{t|p(t)=1\} $,  $ \mathcal{P}^-=\{t|p(t)=-1\} $, $ n_{+}\triangleq |\mathcal{P}^+| $, and $ n_{-}\triangleq |\mathcal{P}^-| $. Since \eqref{sum} can be expressed as sum of i.i.d variables for large values of $ n_{+} $ and $ n_{-} $ , then, using the central limit theorem, we can write \eqref{sum} as linear combination of two Gaussian distributions as $\tilde{y}=Y_1-Y_2, $ 
	where 
	\begin{align}
		Y_1 \sim \mathcal{N}\left(\frac{\mu}{\beta n},\frac{n_{+}}{\beta^2 n^2}\sigma^2\right),\,\,\,Y_2 \sim \mathcal{N}\left(\frac{\mu}{\beta n},\frac{n_{-}}{\beta^2 n^2}\sigma^2\right).
	\end{align}
	Since $ \tilde{y} $ is a linear combination of two independent Gaussian distributions, then we have:
	\begin{align}
		\tilde{y}& \sim \mathcal{N} \left(\frac{\mu}{\beta n}-\frac{\mu}{\beta n},\frac{n_{+}}{\beta^2 n^2}\sigma^2+\frac{n_{-}}{\beta^2 n^2}\sigma^2\right),\nonumber\\
		\tilde{y} &\sim \mathcal{N} \left(0,\frac{1}{\beta^2 n}\sigma^2\right).
	\end{align}
	Now, we can show that $ \tilde{b} $ is a Gaussian variable since it is a summation of a constant value with a Gaussian variables: 
	\begin{align}
	\tilde{b} \sim \mathcal{N}\left(E(\tilde{b})=b,\sigma^2_{\tilde{b}}=\frac{\sigma^2}{\beta^2 n}\right).
	\end{align}
	Then, to analyze the probability of error we consider $ b=1 $. In this case an error occurs when $ \tilde{b}<0 $, therefore, the probability of an error is:
	\begin{align}
	\textrm{Pr}\left\{\tilde{b}<0|b=1\right\}&=\frac{1}{2}\textrm{erfc}\left(\frac{E(\tilde{b})}{\sigma_{\tilde{b}}\sqrt{2}}\right)=\frac{1}{2}\textrm{erfc}\left(\frac{\beta\sqrt{n}}{\sigma \sqrt{2}}\right).\label{errprob}
	\end{align}
	The same error probability can be obtained for $ b=-1 $. 
\end{proof}
From \eqref{errprob} we can observe that for large values of $ \beta $ and $ n $, the bit extraction error goes to zero. However, selecting large values for $ \beta $ and $ n $ will cause some latency and computational challenges for IoT devices which will be discussed later.
\subsection{Static Watermarking for Attack Detection in IoTD}
Now, using the SS method we present a technique for authentication of the signals transmitted from an IoTD to the cloud. We first generate a random pseudo-noise binary sequence with $ n $ samples. Also, for every IoTD, we define a bit stream $ s $ with $ n_{s} $ samples. Then using \eqref{watermarking}, we embed every bit of $ s $ in $ n $ samples of $ y $. Therefore, for any bit stream $ s $, we use $ nn_s $ samples of $ y $, and this embedding procedure will repeat every $ nn_s $ samples of $ y $. At the cloud, using \eqref{correlation}, we extract the bit stream. In case of a cyber attack, the received signal in the cloud will be $ \bar{y}(t) $ rather than $ w(t) $, and, hence, the extracted bit stream will differ from $ s $. Thus, the cloud will trigger an alarm for declaring the existence of a cyber attack. Fig. \ref{fig:static} shows the block diagram of static watermarking for attack detection in an IoTD.
\begin{figure}[!t]
	\centering
	\includegraphics[width=0.73\columnwidth]{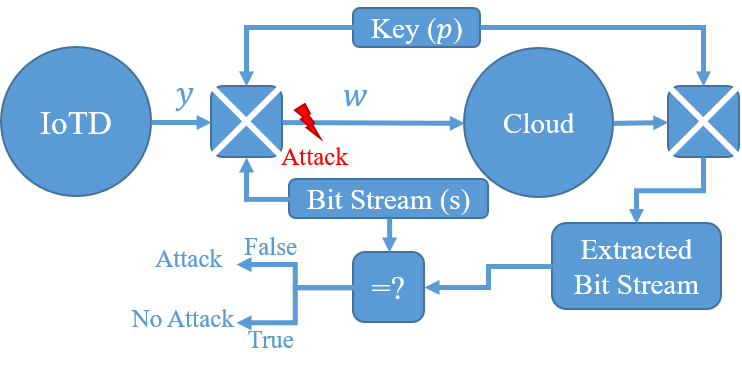}
	\vspace{-3mm}
	\caption{Static watermarking for attack detection.}
	\label{fig:static}
	\vspace{-5mm}
\end{figure}
\vspace{-0.7mm}

In our proposed watermarking scheme, $ \beta $, $ n $, and $ n_s $ play crucial roles in the security of a given IoTD. The value of $ \beta $ must be very smaller than the value of $ \sigma $. The reason is that, for comparable values of $ \beta $, an attacker can extract the key $ p $ and bit stream $ s $, since $ w(t)\simeq\beta b p(t) $ if $ |\beta b p(t)|= \beta >>|y(t)| $. Therefore, we must choose small values for $ \beta $. However, from Theorem \ref{Theorem:BER}, we know that, for small values of $ \beta $, we will have higher bit error rate. To overcome this issue, we have to increase the length of the pseudo-noise key, $ n $. Although increasing the value of $ n $ will reduce the bit error rate, for large values of $ n $, the bit extraction procedure in \eqref{correlation} will result in higher computation load and will also cause higher latency since the cloud must wait for $ nn_s $ samples from the IoTD to detect the attack. Moreover, large values of $ n_s $ will also cause larger delay in the cloud. Therefore, next, we propose a method to choose suitable values for these three parameters. $ \beta $ must be chosen such that the attacker cannot use $ w $ instead of $ p $ to extract the embedded bit. Therefore, we have to adjust $ \beta $ to cause a high bit error rate during the extraction of the hidden bit using the $ w $ sequence. Thus, for two different watermarked signals $ w_1 $ and $ w_2 $ with equal embedded bit $ b=1 $, we have:
\begin{align}
	\hat{b}&=\frac{<w_1,w_2>_{n}}{\beta^2 n^2}=\frac{<y_1+\beta b p,y_2+\beta b p>}{\beta^2 n^2}\nonumber\\
	&=\frac{<y_1,y_2>+<y_1,\beta b p>,<\beta b p,y_2>,<\beta b p,\beta b p>}{\beta^2 n^2}\nonumber\\
	&=Z_1+Z_2+Z_3+b^2,
\end{align}
where $ Z_1 $, $ Z_2 $, and $ Z_3 $ are Gaussian random variables with distributions $ \mathcal{N}\left(\frac{\mu_1}{\beta^2 n^2},\frac{\sigma_1^2}{\beta^4n^3}\right) $, $ \mathcal{N}\left(0,\frac{\sigma^2}{\beta^4 n^3}\right) $, $ \mathcal{N}\left(0,\frac{\sigma^2}{\beta^4 n^3}\right) $, respectively (the proof is analogous to Theorem \ref{Theorem:BER}). $ \mu_1 $ and $ \sigma^2_1 $ are the mean and variance of the multiplication of random variables $ y_1(t) $ and $ y_2(t) $. Then, we can write $ \hat{b} $ as follows:
\begin{align}
	\hat{b}=b^2+Z_4=b+Z_4,
\end{align}
where $ Z_4 \sim N\left(\frac{\mu_1}{\beta^2n^2},\frac{\sigma_1^2+2\sigma^2}{\beta^4n^3}\right) $. Therefore, the bit error rate incurred during the extraction of $ \hat{b} $ will be:
\begin{align}
\textrm{Pr}\left\{\hat{b}<0|b=1\right\}&=\frac{1}{2}\textrm{erfc}\left(\frac{E(\hat{b})}{\sigma_{\hat{b}}\sqrt{2}}\right)\nonumber\\&=\frac{1}{2}\textrm{erfc}\left(\frac{(1+\frac{\mu_1}{\beta^2n^2})\beta^2n\sqrt{n}}{ \sqrt{2(\sigma_1^2+2\sigma^2)}}\right).\label{errprob2}
\end{align}
Since we want to have a high bit error rate for the attacker, we can write:
\begin{align}
	\textrm{Pr}\left\{\hat{b}<0|b=1\right\}&\geq1-\underbar{P},\nonumber\\
	\frac{1}{2}\textrm{erfc}\left(\frac{(1+\frac{\mu_1}{\beta^2n^2})\beta^2n\sqrt{n}}{ \sqrt{2(\sigma_1^2+2\sigma^2)}}\right)&\geq1-\underbar{P}, \label{eq1}
\end{align}
where $ \underbar{P} $ is our desired probability of unsuccessful attack. Moreover we want to have a small extraction error for cloud as in \eqref{errprob}. Thus, we have: 
\begin{align}
		\textrm{Pr}\left\{\tilde{b}<0|b=1\right\}\leq \bar{\text{P}},\,\,\Rightarrow
		\frac{1}{2}\textrm{erfc}\left(\frac{\beta\sqrt{n}}{\sigma \sqrt{2}}\right)\leq \bar{\text{P}}, \label{eq2}
\end{align}
where $ \bar{\text{P}} $ is our desired bit extraction error probability. Therefore, from \eqref{eq1} and \eqref{eq2} we can derive the values for $ \beta $ and $ n $ which satisfy the security and performance requirement of the proposed watermarking scheme. Now, to analyze the attack detection delay, if we consider $ d $ as the acceptable delay in seconds for attack detection, we can find the value of $ n_s $:
\begin{align}
	n_s \leq \frac{df_s}{n} \label{eq3}.
\end{align}
Using \eqref{eq1}, \eqref{eq2}, and \eqref{eq3} we can find the values for the three parameters which satisfy our performance and delay constraints. The proposed SS watermarking method can detect the cyber attacks which only have the ability to change the transmitted data from an IoTD to the cloud. By choosing optimal values for the three parameters of watermarking, the cloud can identify the reliability of the transmitted attack. 

Now, consider a case in which an attacker can also collect data from the IoTDs. In this case, the attacher can launch more complex attacks by \emph{eavesdropping} the legitimate transmitted data from the IoTD. For instance, if the attacker starts to record the transmitted data from the IoTD and sum the recorded data for a long time, it can potentially reveal the key $ p $. As an example, if the attacker collects the data for $ m $ windows of size $ nn_s $ and adds this data together, it will end up having the following signal:
\begin{align}
\bar{w}_m(t)=\sum_{i=1}^m y_i(t) +m\beta b p(t),
\end{align}
where $ y_i $ is the signal received in window $ i $ from an IoTD, and $ \bar{w}_m $ is the summation of collected data. If we consider $ \sigma^2_m $ as the variance of the sum of $ m $ random variable $ y_1(t),\dots,y_m(t) $ then there will exists a value for $ m $ where $ m^2\beta^2 >> \sigma^2_m $. Therefore, the attacker can use $ \bar{w}_m \simeq m\beta b p $ as the key for watermarked signal. The reason for the success of this eavesdropping attack is due to embedding a static bit stream $ s $ in all the windows. However, if the bit stream $ s $ changes dynamically in each window of $ nn_s $ samples then, the eavesdropping attack will not succeed anymore. In the following section, we propose a dynamic bit stream generation using a deep learning method.
\section{Deep Learning for Dynamic IoT Watermarking}\label{sect:DeepMark}
To improve our security scheme, we propose a novel deep learning watermarking method for dynamically generating the bit stream $ y $ which can thwart eavesdropping attacks. In our new dynamic watermarking scheme, we use the fingerprints of the signal $ y $ generated by an IoTD to update the bit stream $ s $, dynamically. Signal fingerprints can be seen as unique identifiers of a signal that can be mapped to a bit stream. Signal processing methods such as spectral flatness, central moments, skewness, and  kurtosis can be used for extraction of fingerprints from signals \cite{audio2002Haitsma,Wavelet2012Bertoncini,Wavelet1995Learned,Mellin2012Harley}. Here, we use a deep learning framework to extract these kind of stochastic features the IoTD signals. Deep learning has been successfully applied in a wide range
of areas with significant performance improvement, including
computer vision, natural language processing, and speech recognition \cite{chen2017machine}. One of the widely-used deep learning methods for sequence classification is called LSTM \cite{chen2017machine} and \cite{speech2013Graves}. In the following, we explain how we use LSTM to extract the fingerprints from the IoTD signals.
\subsection{LSTM for Signal Fingerprinting in IoTD}
To dynamically extract fingerprints from IoTD signals, we propose an LSTM algorithm that allows an IoTD to update the bit stream based on the sequence of generated data. 

LSTMs are one of the deep recurrent neural networks (RNNs) that can store information for long periods of time and thus can learn long-term dependency of a given sequence \cite{speech2013Graves}. Essentially, LSTMs processes an input $ \left(y(1),\dots,y(nn_s)\right) $ by adding new information into a memory, and using gates which control the extent to which new information should be memorized, old information should be forgotten, and current information should be used. Therefore, the output of LSTMs will be impacted by the network activation in previous time steps and thus LSTMs are suitable for our IoT application in which we want to extract fingerprints from signals which are dependent to previous time steps and LSTMs are a sequence to sequence mapping.

\begin{figure}[!t]
	\centering
	\includegraphics[width=0.8\columnwidth]{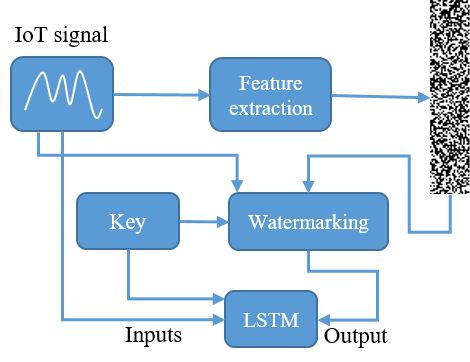}
	\caption{Training phase of an LSTM in an IoTD.}
	\label{fig:training}
	\vspace{-5mm}
\end{figure}
During the training phase, the parameters of the algorithm are learned from a given training dataset of different IoTD such as accelerometer, gyroscope and positioning devices. As in \cite{audio2002Haitsma,Wavelet2012Bertoncini,Wavelet1995Learned,Mellin2012Harley} we choose spectral flatness, mean, variance, skewness, and kurtosis as features that are extracted from a signal with length $ nn_s $ and then map these values to a bit stream with length $ n $. Next, we watermark this extracted bit stream into the original signal using a key. To train the LSTM, we use the original signal $ y $ and pseudo noise key $ p $ as input stream and the watermarked signal $ w $ as output stream. Fig. \ref{fig:training} shows the training phase using a block diagram model. Next, we illustrate how we use the trained LSTM to dynamically watermark an IoT signal.

\subsection{LSTM for Dynamic IoTD Watermarking}
At the cloud, we use an LSTM for bit extraction. To train this LSTM, we use the watermarked signal $ w $ and key $ p $ as inputs to the neural network and the features of the original signal and extracted bit stream as outputs. The block diagram model for training phase of LSTM in the cloud is shown in Fig. \ref{fig:trainingcpu}. Using the LSTM block of Fig. \ref{fig:training} at the IoTD and the LSTM block of Fig. \ref{fig:trainingcpu} at the cloud, we propose a dynamic LSTM watermarking scheme to implement an attack detector at the cloud. 

In this method, a predefined bit stream is not used since a bit stream is dynamically generated inside the LSTM blocks at the IoTD and the cloud. This dynamic bit stream generation at the hidden layers of LSTMs solves the eavesdropping attack problem, since recording and summing the IoTD signals will not increase the power of the key sequence and the attacker will not be able to extract the key and bit stream. Using this method, the IoTD inserts the generated signal and key in its LSTM block in each window of $ nn_s $ samples and produces a watermarked signal with different bit stream $ s $ in each window. At the cloud, the received watermarked signal and key are passed from the LSTM. Then, the two outputs (the extracted bits, and extracted features) are compared. In the case of dissimilarity of two sequences, the attack alarm is triggered. Fig. \ref{fig:dynamicwatermarking} shows the block diagram of dynamic LSTM watermarking for attack detection.
\begin{figure}[!t]
	\centering
	\includegraphics[width=0.75\columnwidth]{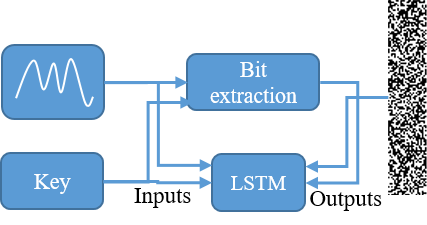}
	\caption{Training phase of an LSTM at the cloud.}
	\label{fig:trainingcpu}
	\vspace{-5mm}
\end{figure}
\begin{figure}[!t]
	\centering
	\includegraphics[width=\columnwidth]{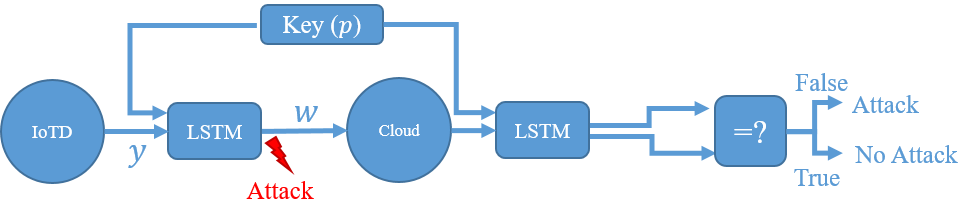}
	\caption{Dynamic watermarking for attack detection.}
	\label{fig:dynamicwatermarking}
	\vspace{-7mm}
\end{figure}
\section{Simulation Results and Analysis}\label{sect:sim}
For our simulations, we use a real dataset from an accelerometer with sampling frequency $ f_s=1 \text{kHz} $. In each simulation, we derive the optimal values for $ \beta $, $ n $, and $ n_s $ using the method proposed in Section \ref{sect:sensor} such that they satisfy the reliability and delay constraints.

Fig. \ref{fig:trainingout} shows the output of the LSTM trainer with $ n=25,n_s=25, $ and $ \beta=0.5 $. It can be seen from Fig. \ref{fig:trainingout} that the trained output of the LSTM, which is the dynamic watermarked signal, is very close to the training target $ w $. Moreover, Fig. \ref{fig:performance} shows the performance of training, in which the training converges after 269 \emph{epochs}. Here, an epoch is a measure of the number of times all the training vectors are used once to update the weights of the neural network. The training error is 0.0055 which is calculated using mean squared error. Moreover, we tested the trained LSTM on another accelerometer data, and the testing error is close to 0.02 which is acceptable for our IoT application.
\begin{figure}[!t]
	\centering
	\includegraphics[width=0.9\columnwidth]{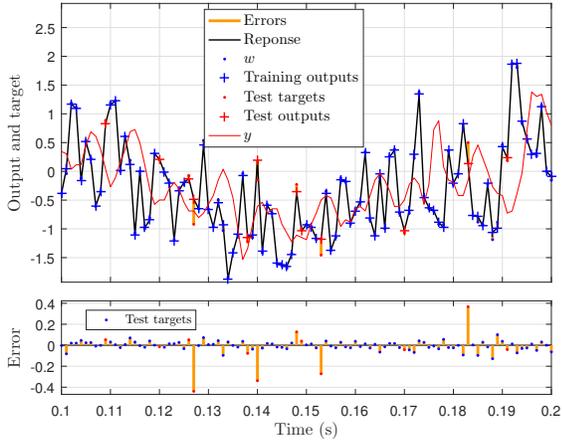}
	\caption{Comparison of IoT signal, static watermarked signal, dynamic watermarked signal, and train error.}
	\label{fig:trainingout}
	\vspace{-5mm}
\end{figure}
\begin{figure}[!t]
	\centering
	\includegraphics[width=0.95\columnwidth]{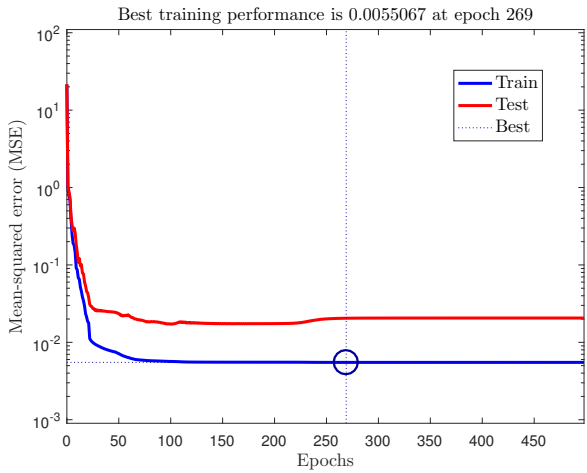}
	\caption{Training performance.}
	\label{fig:performance}
	\vspace{-5mm}
\end{figure}
\begin{figure}[!t]
	\centering
	\includegraphics[width=0.9\columnwidth]{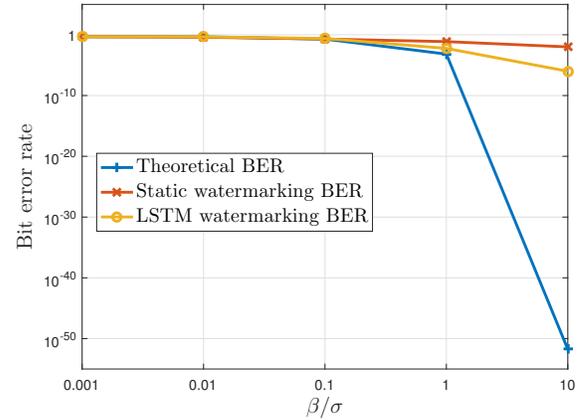}
	\caption{Bit error rate of proposed watermarking algorithms.}
	\label{fig:beta}
	\vspace{-5mm}
\end{figure}
\begin{figure*}
	\centering
	\includegraphics[width=\textwidth]{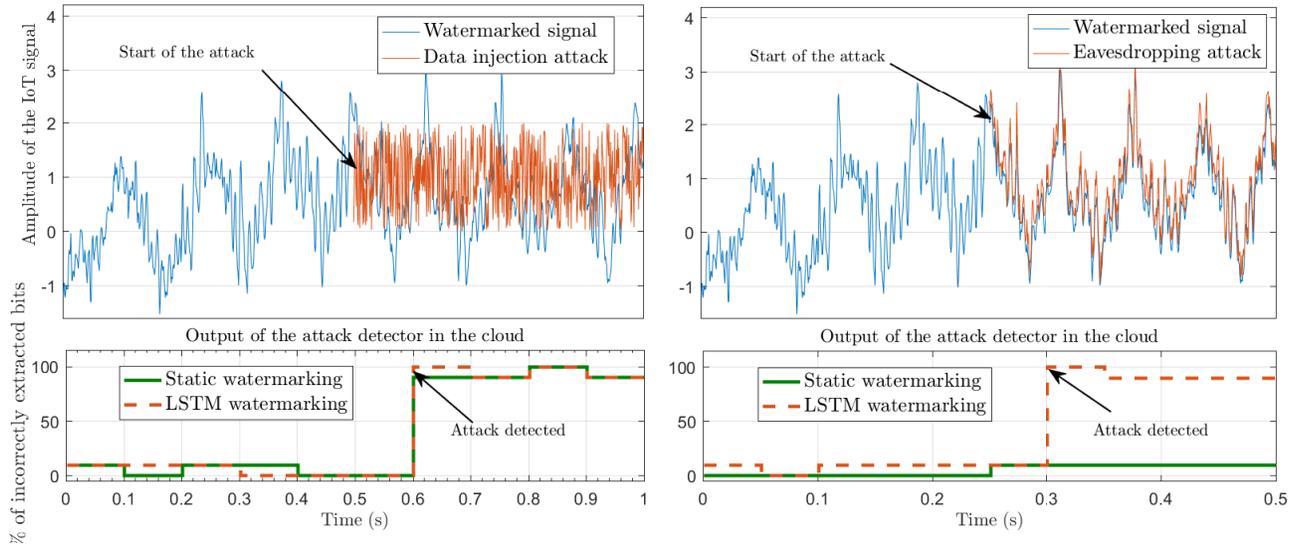}
	\caption{Attack detection analysis of static and LSTM watermarking schemes. The  y-axis in the lower figures show the percentage of difference between the extracted bit stream and the hidden bit stream. This difference is used at the cloud as a metric to detect attacks on the IoTDs.}
	\label{fig:attackdetection}
	\vspace{-5mm}
\end{figure*}
\begin{figure}[!t]
	\centering
	\includegraphics[width=0.9\columnwidth]{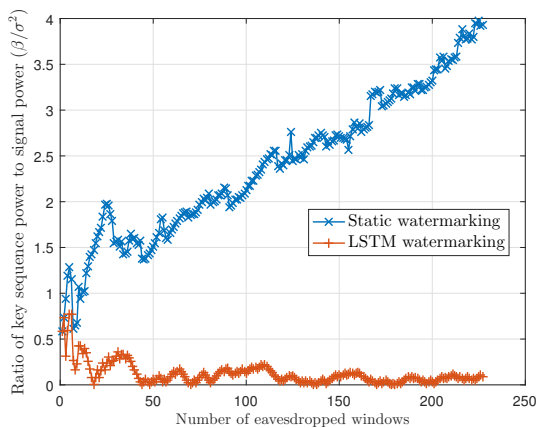}
	\caption{The ratio of key power to signal power in eavesdropping attack.}
	\label{fig:eavesdroppingpower}
	\vspace{-7mm}
\end{figure}

Fig. \ref{fig:beta} illustrates the higher performance of LSTM compared to to static watermarking in bit extraction. From \eqref{errprob}, we know that higher $ \beta/\sigma $ results in lower bit error. We can see from Fig. \ref{fig:beta} that the extraction error rate for LSTM is approximately $ 1 $ order of magnitude lower than the static watermarking when $ \beta/\sigma=1 $. This ratio gets better for higher $ \beta/\sigma $, since we can observe that the error rate of LSTM is almost $ 2 $ orders of magnitude lower than static watermarking when $ \beta/\sigma=10 $. This result allows designing attack detectors with lower delay, since we can choose lower $ n $ for LSTM which results in smaller window size and reduces the detection delay.

To analyze the functionality of our proposed watermarking schemes in attack detection, we choose a static watermarking block with $ n=10, n_s=10$ and $ \beta=0.5 $. We also train our LSTM block with these features. Then, we implement two types of attack to the signal: a) a data injection attack in which the attacker starts to change the IoT device signal, and b) an eavesdropping attack, in which the attacker records the data from the IoT device, extracts the bit stream and implements an attack with the the same watermarking bit stream. In Fig. \ref{fig:attackdetection}, the attack detectors compare the extracted bit stream with the actual hidden bit stream and the percentage of difference between these two is considered as a metric for attack detection. In other words, high difference between the two bit streams, an attack detection alarm is triggered. Fig. \ref{fig:attackdetection} shows that, for the first attack, both watermarking schemes can detect the attack. However, for eavesdropping, static watermarking cannot detect the existence of attack while the LSTM performs well. The reason is that in static watermarking the bit stream is the same for all the time windows, while in LSTM the watermarked bit stream dynamically changes for each time window. 

Fig. \ref{fig:eavesdroppingpower} illustrates how an eavesdropping attack operates against the two watermarking schemes. We can see that, in static watermarking, the attacker records the signal and by summing the recorded data of each window, increases the ratio of pseudo-noise key power to the signal power and extracts the bit stream. However, since in LSTM, the bit stream dynamically changes in each window, the summation of the recorded data will not increase the ratio of the key power to the signal power. Therefore, the attacker will not be able to extract the bit stream and key from the recorded data. In addition, we can see from Fig. \ref{fig:attackdetection} that the delay of attack detection is $ 0.1 $ seconds since the attacks starts at $ 0.5 s $ and the attack detector triggers the alarm at $ 0.6 s $. The reason is that, for the window size of $ n\times n_s /f_s=0.1 $ seconds, the cloud must wait for one time window to collect the data from the IoTD.
\section{Conclusion}\label{sect:conc}
In this paper, we have proposed a novel deep learning method based on LSTM blocks for enabling attack detection of data injection and eavesdropping in IoT devices. We have introduced two watermarking schemes in which the IoT's cloud, which collects the data from the IoT devices, can authenticate the reliability of received signals. We have shown that our proposed LSTM method is suitable for IoT security due to low complexity, small delay, and high accuracy in attack detection. Simulation results have shown that dynamic LSTM watermarking can also detect existence of complicated attacks such as eavesdropping in which the attacker collects data from  the IoT devices and designs an undetectable attack.
\def\baselinestretch{0.89}
\bibliographystyle{IEEEtran}
\bibliography{references}
	
	%\vfill
	


\end{document}